\documentclass[5p]{elsarticle}
\usepackage[title]{appendix}
\usepackage[section]{placeins}
\usepackage{amssymb}
\usepackage{amsmath}
\usepackage{graphicx}
\usepackage{url}
\usepackage{subfig}
\usepackage{algorithm2e}
\urldef{\mailsa}\path|{alfred.hofmann, ursula.barth, ingrid.haas, frank.holzwarth,|
\urldef{\mailsb}\path|anna.kramer, leonie.kunz, christine.reiss, nicole.sator,|
\urldef{\mailsc}\path|erika.siebert-cole, peter.strasser, lncs}@springer.com|    

\newtheorem{theorem}{Theorem}
\newtheorem{lemma}[theorem]{Lemma}
\newtheorem{corollary}[theorem]{Corollary}
\newdefinition{rmk}{Remark}
\newproof{proof}{Proof}
\newproof{proofoftheorem}{Proof of Theorem \ref{thm2}}
\newtheorem{definition}{Definition}[section]

\begin{document}

\title{Hardness of Modern Games}

\author{Diogo M. Costa}
\author{Alexandre P. Francisco}
\author{Lu\'is M. S. Russo}
\address[associate]{INESC-ID and the Computer Science and Engineering
    Department of Instituto Superior T\'ecnico, Universidade de
    Lisboa.}

\begin{abstract}

We consider the complexity properties of modern puzzle games, Hexiom, Cut the Rope and Back to Bed. The complexity of games plays an important role in the type of experience they provide to players. Back to Bed is shown to be PSPACE-Hard and the first two are shown to be NP-Hard. These results give further insight into the structure of these games and the resulting constructions may be useful in further complexity studies.

\end{abstract}

\maketitle

\section{Introduction}
Not long after the introduction of the notion of reducibility between computational problems~\cite{karp2,simple_npc}, it was applied to games, especially in the 2-player realm~\cite{checkerspspace,chess}, leading to a generalization of the non-deterministic model of computation~\cite{alternation}. More recently, it has also proven fruitful to study single-player games, whether pen-and-paper puzzle games~\cite{asp,fillmat} or classic video games~\cite{lemmingscormode,nintendopaper}.

In this paper, we study three modern single-player puzzle games using NP and PSPACE reductions. We obtain the following results:

\begin{itemize}
    \item Hexiom is NP-Complete (Section~\ref{proof:hexiom}) by reducing from CircuitSAT (satisfiability of boolean circuits). Hexiom uses a board made up of an hexagonal grid of cells. Pieces are placed in cells. Fixed pieces define the original puzzle and can not be moved; they may or may not be numbered. Movable pieces define the player's hand, and must be placed on the board according to the following constraint: every piece numbered with $n$ must have exactly $n$ adjacent numbered pieces. We start off by building simple gadgets to represent wires, NOT and NOR gates that will compose the boolean circuit. We finish the proof by analyzing the parity (a property relating the sum of all the numbers of the pieces and the existence of a solution) of the puzzle and adding parity gadgets to ensure that every solution of CircuitSAT instance has a corresponding solution in our Hexiom instance.
    \item Cut the Rope is NP-Hard (Section~\ref{proof:cuttherope}) by reducing from 3-CNFSAT. Cut the Rope is based on the physics of circles and ropes, where the goal is to move a circle (candy) to a certain position in the level (Om Nom). In the process, the circle must avoid falling off the level or colliding with spikes, and may be teleported between pairs of points (by hats) or float (in a bubble). We build several gadgets to represent formula components and speed adjusting. 
    \item Back to Bed is PSPACE-Hard (Section~\ref{proof:backtobed}) by reducing from TQBF (True Quantified Boolean Formula). Back to Bed is a labyrinth like game, where the goal is to maneuver a sleepwalking avatar back to bed. In the process, the player controls a different in-game avatar and is able to move obstacles and bridges to restrict the sleepwalker's movements. The sleepwalker must avoid falling off the level and colliding with patrolling dogs. We build gadgets to represent formula components; the gadgets built in this one are more robust as they may need to be traversed several times in the solution. There are also additional gadgets to keep the player and the avatar close together and prevent obstacles from being moved from one gadget to another.
\end{itemize}

The two last proofs also reinforce the fact that the 3-CNFSAT and TQBF frameworks (see Sections \ref{fw:cnf} and \ref{fw:cnf}, or~\cite{nintendopaper}), which are based on the existence of a door mechanism and aimed at platformer games (like Super Mario Bros.), can be applied to games that are neither platformer games nor have doors as an explicit mechanism.

We first motivate the use of each framework by describing the rules of the game, followed by the description of the framework, and finally the reduction itself. More details and proofs of hardness can be found either in the original papers (see Section \ref{section:rw}), or the on the thesis~\cite{mythesis} (including an example of an instance of the original problem and the respective game level, for each game).

\paragraph{Gadget} Every reduction presented in this paper relies on the concept of gadget, introduced here. Gadgets are the building blocks of the reduction, implementing a particular feature of the problem we're reducing from using the rules of the target game. The final reduction is simply a composition of gadgets. The use of gadgets makes it easier to treat the reduction modularly, with its global connectivity illustrated on the frameworks described below.

An example of gadget would be a portion of a game level implementing a clause (in the case of 3-CNFSAT or TQBF), which is typically made up of smaller gadgets that implement locks or doors. In the case of CircuitSAT, we must construct a gadget for each boolean gate, as well as signal propagating wires.

\section{Hexiom}\label{proof:hexiom}
Hexiom is a puzzle game freely playable on the web~\footnote{\url{https://www.kongregate.com/games/moonkey/hexiom}}. The game has received some attention in relation to SAT, but in the opposite direction as the one treated here: Hexiom was reduced to SAT to be solved with SAT solvers~\cite{hexiomsat,hexiomsolve}, whereas we are are reducing SAT to Hexiom, to solve instances of SAT by solving Hexiom puzzles.

\subsection{Rules} The game consists of two groups of hexagonal pieces: one group is the board, which contains empty cells and fixed pieces (numbered or not) on the hexagonal grid; the other group is the player's ``hand'' (also referred to as set $M$ throughout the section), which consists of numbered cells that must be placed on the empty cells of the board. Numbered pieces are locally constrained such that a piece numbered with $n$ has exactly $n$ \textit{numbered} cells adjacent to it. A solution for a given Hexiom puzzle is one where every piece of the player's hand is placed on the board, and every numbered piece has its constraint satisfied. An example of a puzzle is shown in Figure \ref{fig:hexiom_example}. Black hexagons represent fixed pieces that can't be moved; gray hexagons are the empty board cells; red hexagons are the pieces from the player's hand (not shown in Figure \ref{fig:hexiom_example_empty}). 

\begin{definition}
An $n$-piece is a piece numbered with $n$.
\end{definition}

\begin{figure}[!ht]
\center
 \subfloat[Empty puzzle.\label{fig:hexiom_example_empty}]{%
   \includegraphics[scale=0.4]{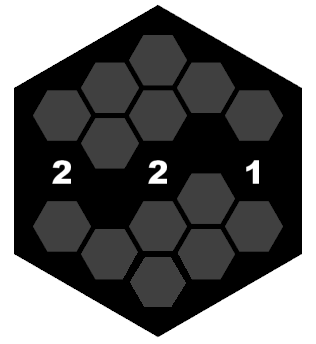}
 }
 \subfloat[Complete puzzle\label{fig:hexiom_example_filled}]{%
   \includegraphics[scale=0.4]{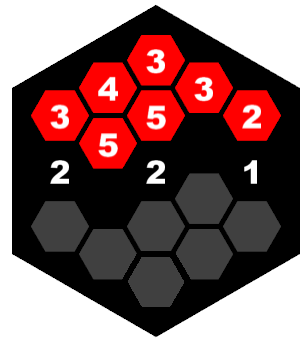}
 }
 \caption{Hexiom puzzle example.}
 \label{fig:hexiom_example}
\end{figure}

\subsection{Reduction from CircuitSAT}\label{fw:circuit}
\begin{definition} 
CircuitSAT is the problem of whether a given boolean circuit has a variable assignment that makes its output $true$.
\end{definition}
The reduction consists in a polynomial time algorithm that, given a specific boolean circuit, creates an Hexiom puzzle such that every solution of the puzzle can be transformed into a solution of the circuit, also in polynomial time. We require gadgets to: choose binary variable values (signals), propagate them through wires, and a set of functionally complete boolean gates.

The problem of crossover (of wires or paths in games), which requires additional gadgets in the two other frameworks, is already solved for CircuitSAT using three $XOR$ gates (or an equivalent composition of other gates).

Functional completeness is especially important in providing flexibility to the reductions. A set of boolean gates is functionally complete if it is sufficient to simulate any boolean gate. For example, every boolean gate can be expressed as a composition of the sets {$NOR$} or {$NAND$}, making each of them functionally complete. In our experience, a gadget to negate the signal, a $NOT$ gate, is simple to implement, and so any gadget for a gate of the set $\{AND, OR, NOR, NAND\}$ is sufficient to complete the reduction. 

Proofs based on Non-Deterministic Constraint Logic~\cite{ncl} (NCL), differ by using monotone logic; the monotonic property removes the need for $NOT$ gates, but usually requires both $AND$ and $OR$ gates. 

\subsection{Proof overview} We prove Hexiom to be NP-Complete using a reduction from CircuitSAT. The proof has two main parts: one consists in simulating a circuit and is typically the only requirement in CircuitSAT reductions; the other consists in adding extra gadgets to allow the player to place any excess pieces he may have, a condition that does not generally apply to CircuitSAT reductions. Any gray cell can be used to place a piece, but only the cells inside the gadget's borders will be useful (and necessary) to satisfy the circuit component. Outside empty cells can be used to place the excess pairs of 1-pieces.

\subsection{Circuit simulation} As described above, we start by replacing each component of a circuit with the respective gadget. In this proof, each gadget consists of one fixed group of pieces added to the board, and another of pieces added to the player's hand. The latter must allow for all the valid configurations of the gadget. Figure \ref{fig:hexiom_gadgets} shows the fixed part of every gadget. Table \ref{table:hexiom_balance} summarizes the pieces in each of the two groups for every gadget. Figure \ref{fig:hexiom_wireconfig} also shows an extended wire with its two possible configurations.

\begin{figure}[!ht]
\center
 \subfloat[Variable selection and Wire.\label{fig:hexiom_wire}]{%
   \includegraphics[scale=0.33]{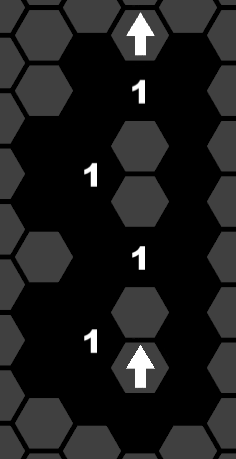}
 }
 \subfloat[Turn.\label{fig:hexiom_turn}]{%
   \includegraphics[scale=0.31]{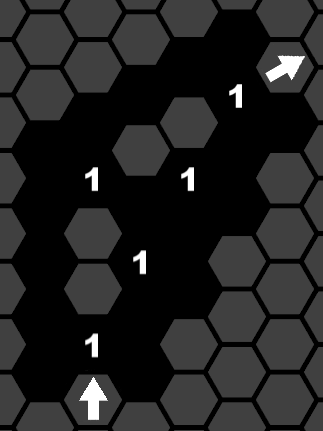}
 }
 \subfloat[Negated FAN-OUT gadget.\label{fig:hexiom_fanout}]{%
   \includegraphics[scale=0.3]{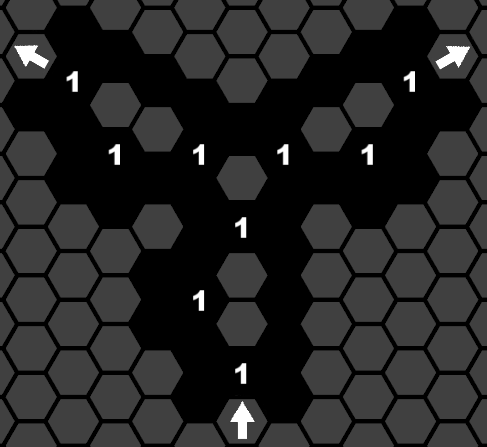}
 }
 \\
 \subfloat[NOT gadget.\label{fig:hexiom_not}]{%
   \includegraphics[scale=0.825]{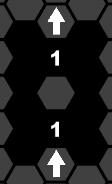}
 }
 \subfloat[NOR gadget.\label{fig:hexiom_nor}]{%
      \includegraphics[scale=0.5]{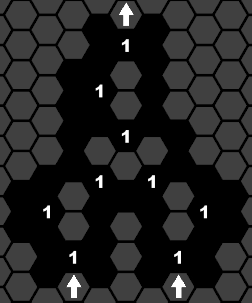}
     }
 \caption{Hexiom gadgets}
 \label{fig:hexiom_gadgets}
\end{figure}

    \paragraph{Wires, Wire extensions, Turns and Variable selection} The wire and variable selection gadgets are shown in Figure \ref{fig:hexiom_wire}. A wire is composed of one or more wire extensions and starts at another gadget's output. A wire extension is a gadget that is 3 cells long (one 1-piece with two empty cells below). Figure \ref{fig:hexiom_wire} shows a wire composed of two wire extensions, while Figure~\ref{fig:hexiom_wireconfig} shows a wire composed of three extensions. Both include variable selection. The variable selection gadget is simply an unnumbered cell below a wire extension. In Figures\ref{fig:hexiom_wire} and \ref{fig:hexiom_wireconfig}, the two bottom empty cells determine the variable's value. An extended wire and signal propagation are shown in Figure~\ref{fig:hexiom_wireconfig}. Turns are simply wire extensions that change the angle of propagation (see Figure \ref{fig:hexiom_turn}).

    \paragraph{Signals} The value of the signal is given by the position of the movable cell relative to the center 1-piece. Placement above the 1-piece represents true, while placement below the 1-piece represents false. Note that this relative position is never broken by wire extensions.
    
    \begin{lemma}The gadget shown in Figure \ref{fig:hexiom_wire} carries its input signal forward, like a wire.\end{lemma}
    \begin{proof} Because of the 1-piece on the side of the wire extension, one movable cell must be placed in one of the two empty spaces (and only one). Because of the center 1-pieces, the movable 1-piece will be placed in the same relative position throughout the whole wire. Because the signal is determined by this relative position, and it is conserved, the signal is correctly propagated.
    \qed \end{proof}
    Some remarks must be made regarding the types of cells placed. In the variable selection (bottom portion of Figure \ref{fig:hexiom_wire}), the bottom-most empty space has a single adjacent numbered cell. As exemplified in Figure~\ref{fig:hexiom_wireconfig}, to select a true value for the variable, a 1-piece must be placed. The empty space above it, however, has two adjacent numbered cells. To select a false value for the variable, a 2-piece must be placed (note that this doesn't happen in the wire \textit{extensions}, where a 2-piece is placed to propagate the signal regardless of the signal). As such, variable selection gadgets, when used in the construction, add one 1-piece and one 2-piece to the set $M$. Wire extensions add only a single 2-piece to $M$.
    The former will always have one excess piece, with which we deal after we describe the circuit gadgets. The latter has no excess piece.
    
    \begin{figure}[h!]
        \centering
        \includegraphics[scale=0.4]{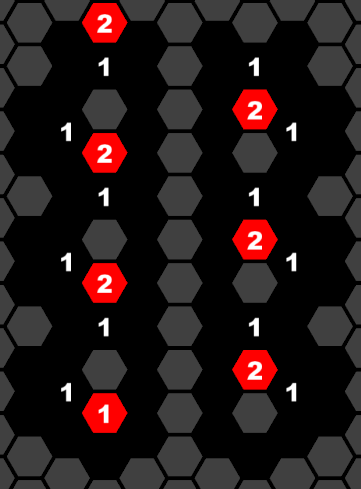}
        \caption{Hexiom's true (left) and false (right) signals through a wire.}
        \label{fig:hexiom_wireconfig}
    \end{figure}

    \paragraph{NOT gadget} The NOT gadget, shown in \ref{fig:hexiom_not}, is introduced before/after any wire extension to negate the signal. NOT gadgets are 2 cells long (one 1-piece with a single empty cell below). This is relevant because, aside from being used to negate a signal, it can also be used to shift the wire's position by using a double negation, 4 cells long. Because the wire extensions' length is always a multiple of 3, alignment problems could arise when connecting the different gadgets.
    
    \begin{lemma}The gadget shown in Figure \ref{fig:hexiom_not} behaves like a NOT gate, reversing its input signal.\end{lemma}
    \begin{proof}If the input signal is false (there is a numbered cell below the bottom-most 1-piece), the output signal will be true (there will be another cell above the top-most 1-piece). If the input signal is true, there will be a 2-piece between the two 1-pieces, and so no cell above the top-most 1-piece, so the output signal is false. 
    
    The gadget will then require a single 2-piece, or no cell at all, depending on the incoming signal. The possible excess 2-piece will be treated after the NOR gate is described. Each NOT gate adds one 2-piece to set $M$.
    \qed \end{proof}
    \begin{lemma}Two NOT gates can be combined into a wire extension of length 4.\end{lemma}
    \begin{proof}The signal passing through the two NOT gadgets will be negated twice, becoming the original signal again. Each NOT gate has a length of 2, so the overall length is 4.
    \qed \end{proof}
    
    This can be used to shift the wires by a single cell, because $gcd(4, 3)=1$. If the wires are misaligned, we can extend one of them with a single wire extension, by 3, and the other with two NOT gates, by 4. This will bring them a single unit closer to alignment. Repeating the process will align them completely.
    
    \paragraph{Negated FAN-OUT gadget} The FAN-OUT gadget multiplies the incoming signal, propagating it in two directions. In our case, it also negates the outgoing signals, but this is not a problem given the existence of the NOT gadget. The Negated FAN-OUT gadget is just the triangular configuration of 1-pieces (see Figure \ref{fig:hexiom_fanout}. The remainder of the gadget simply illustrates how it connects to wire extensions. 
    \begin{lemma} The gadget in Figure \ref{fig:hexiom_fanout} behaves like negated a FAN-OUT gate, having two output wires carrying the opposite signal of its input wire.
    \end{lemma}
    \begin{proof} The gadget has two valid configurations - either a 3-piece in the center, or no cell at all. If there is no cell, the 1-piece constraints are satisfied in the wire extensions with 2-pieces.
    \qed \end{proof}
    
    A single 3-piece is added to set $M$. When the input is true, the 3-piece is in excess (can't be placed in the gadget).
    
    \paragraph{NOR gadget} The NOR gadget (Figure \ref{fig:hexiom_nor}) is similar to the FAN-OUT, in the sense that only a small portion of the gadget, composed of a triangular set of 1-pieces is required to make the gadget. The two extra empty cells suffice to create the NOR behaviour, and the remainder is a straightforward connection to the wire extensions. 
    \begin{lemma} The gadget shown in Figure \ref{fig:hexiom_nor} behaves like a NOR gate, having an output of true when both inputs are false and false otherwise.
    \end{lemma}
    \begin{proof}
    There are 4 valid configurations: a 3-piece in the center when both inputs are true, a 2-piece on the side of the false input when only one input is true, and no cell in the center portion when both inputs are false. 
    
    The gadget adds one 3-piece and one 2-piece to set $M$ to cover all four cases. Only one of the pieces is used for a given configuration, the other being in excess. 
    \qed \end{proof}
    
\subsection{Excess pieces} With the gadgets above, we can construct a level that simulates any single boolean circuit and, if the circuit is satisfiable, the player is able to place pieces in a way that satisfies every constraint of the fixed and placed pieces. However, several pieces from the player's hand remain unplaced, so we don't have a correspondence between the circuit and puzzle solutions just yet. We tackle this problem by first reducing every excess piece to 1-pieces, and then to a single excess 1-piece, effectively reducing it to a parity problem as defined below. Finally, we show how to solve the parity problem by showing how any valid solution to an Hexiom puzzle can be seen as a graph (potentially with multiple components), and how the sum of the degrees guarantees that valid circuit solutions never creates an excess 1-piece.

\begin{definition}
Hexiom's parity problem consists in determining whether, for any otherwise valid configuration of the board, the player is left with a single 1-piece in his hand.
\end{definition}

\paragraph{Reducing excess pieces to the parity problem} We will describe the additional gadgets here. Like the main gadgets described above, every gadget has a fixed component added to the board, and a component added to the player's hand. Each gadget has two valid configurations, one of them being using only pieces belonging to the gadget, and the other using one excess piece from the player's hand. Each gadget may also have excess pieces, but these will exclusively be 1-pieces. Figure \ref{fig:hexiom_excess} shows the two gadgets.

\begin{figure}[!ht]
\center
 \subfloat[\{3\} or \{1,2\} gadget.\label{fig:hexiom_3_1}]{%
   \includegraphics[scale=0.6]{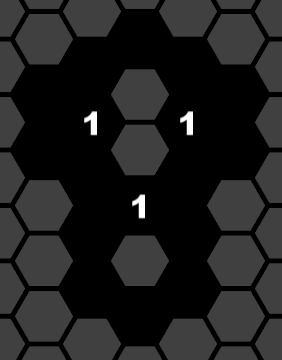}
 }
 \subfloat[\{2,1,1\} or \{1,1\} gadget.\label{fig:hexiom_2}]{%
      \includegraphics[scale=0.4]{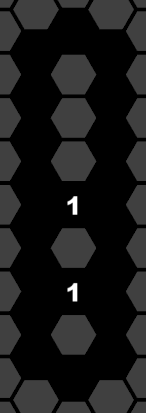}
     }
 \caption{Gadgets to solve Hexiom's excess cells}
 \label{fig:hexiom_excess}
\end{figure}

    \begin{lemma}
    The \{2,1,1\} or \{1,1\} gadget shown in Figure \ref{fig:hexiom_2} can be used to satisfy a 2-piece.
    \end{lemma}
    \begin{proof} The gadget can be satisfied by placing a 1-piece below the bottom 1-piece and another 1-piece above the top 1-piece. Alternative, a single 2-piece can be placed between the two 1-pieces, and two 1-pieces are placed adjacent to each other at the top of the gadget.
    \qed \end{proof}
    
    Two 1-pieces are added to the set $M$. We add one of these gadgets for each NOT gate, NOR gate, and \{3\} or \{1,2\} gadget described below.

    \begin{lemma}
    The \{3\} or \{1,2\} gadget shown in Figure \ref{fig:hexiom_3_1} can be used to satisfy one 3-piece at the cost of one 1-piece and one 2-piece.
    \end{lemma}
    \begin{proof} A 2-piece may only be placed on the top-most free cell of the gadget, adjacent to two 1-pieces. Then, the bottom cell must be satisfied by placing another 1-piece at the bottom of the gadget. Alternatively, a single 3-piece may be placed in the center of the gadget, satisfying the three adjacent 1-pieces.
    \qed \end{proof}
    
    To make use of this gadget, we add a 1-piece and a 2-piece to $M$. We add one gadget for each NOR and FAN-OUT gadgets. In this way, given the previous gadget, we can replace a 3-piece with a 1-piece whenever we get a surplus 3-piece either of the two gates.
    
    Note that because this \{3\} or \{1,2\} gadget may have an excess 2-piece, we add one \{2,1,1\} or \{1,1\} gadget for every \{3\} or \{1,2\} gadget too.

\paragraph{Solving parity} At this point, we have an arbitrary number of excess 1-pieces. The first step is to realize that two 1-pieces can be satisfied simply by being placed adjacent to one another. As a result, we can place every pair of 1-piece in any empty board space. We are left with, at most, one excess 1-piece. 

One possible solution is to double the whole construction. If the player has a satisfying configuration, choosing the same assignment on both copies will result in $(2\times0)$, or $(2\times1)$ 1-pieces, which can be easily solved, since the number is even.

We present a different solution that, although not as straightforward, will incur almost no blow-up in size: only a single fixed 1-piece is added, at most.
We argue that this extra 1-piece is configuration-independent, that is, it depends only on the puzzle, and not the variable assignments.

    Our argument rests on the following observation:
    \begin{lemma}\textit{Every connected and valid Hexiom configuration corresponds to an instance of a graph $G$ (with a set of vertices $V$ and a set of edges $E$) where each numbered cell becomes a vertex, every adjacent numbered cell shares an edge and the number of each cell corresponds to the degree of its respective vertex.}
    \end{lemma}
    \begin{proof} This observation is illustrated in Figure \ref{fig:hexiom_graph} for the previously shown example in Figure \ref{fig:hexiom_example}.  In a valid Hexiom configuration, each $n$-piece is satisfied if and only if it has exactly $n$ adjacent cells. On the graph, this is the number of edges of a vertex, its degree.\qed \end{proof}
    
    This is relevant because we know that: $$\sum_{v \in{V}}^{} degree(v) = 2 \times \left\vert{E}\right\vert$$ From this, we can conclude that the sum of every cell's number, in valid a configuration, must be even.
    
    \begin{corollary}Every valid Hexiom configuration is an instance of a graph where the number of each cell corresponds to the degree of its respective vertex.
    \end{corollary}
    This follows simply because a valid configuration is composed of all the valid connected configurations. The overall sum of the cells' numbers will be a sum of even numbers, which is also an even number.
    
    \begin{figure}[!ht]
     \center
     \includegraphics[scale=0.3]{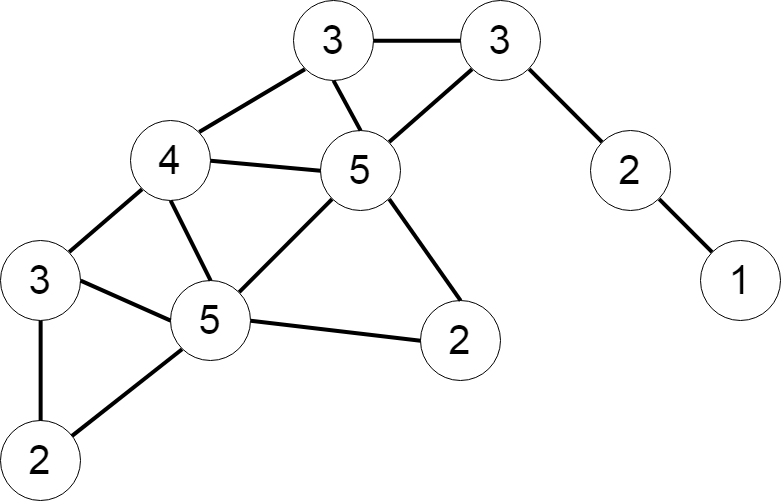}
     \caption{Graph of example in Figure \ref{fig:hexiom_example}}
     \label{fig:hexiom_graph}
    \end{figure}
    
    The whole construction must then have an even sum of degree.
    We then have a global property for the puzzle that, for some $k \in{\mathbb{N}}$: $$\sum_{piece \in{F}\cup{M}}^{} number(piece) = 2 \times k$$
    
    Each gadget, as described, has a constant sum for $F \cup M$. This means that the final construction, too, will have a constant sum, as all the elements of $M$ must be placed on the board. This solves our parity problem, because we can simply sum the numbers of the pieces of every gadget and check whether the result is even or odd, independently of the specific configuration. If the number is odd, we know that we will always end up with an extra 1-piece once we find an otherwise valid configuration (that must sum to even), and so we can add a single fixed 1-piece to satisfy the excess piece. Conversely, if it sums to even, we will always end up with an even number of 1-pieces, so no additional gadget is needed.

    To conclude, the following table shows the balance of each gadget in terms of fixed and free pieces to make the final accounting easier. In the end, the overall sum must be even.
    \begin{center}
    \begin{tabular}{ |c|c|c|c| }
    \hline
    \label{table:hexiom_balance}
     Gadget & Fixed & Free & Total  \\ [0.5ex] 
     \hline \hline
     SELECT & \{1,1\} & \{1,2\} & 5 \\ 
     WIRE & \{1,1\} & \{2\} & 4 \\
     NOT & \{1\} & \{2\} & 3 \\ 
     FAN-OUT & \{1,1\} & \{3\} & 5 \\
     NOR & \{1\} & \{2, 3\} & 6 \\
     2-EXCESS & \{1,1\} & \{1,1\} & 4 \\
     3-or-\{2,1\} & \{1,1,1\} & \{2,1\} & 6 \\
     \hline
    \end{tabular}
    \end{center}

    \begin{theorem}
        Hexiom is NP-Complete.
    \end{theorem} 
    \begin{proof} We show that CircuitSAT $\le _P$ Hexiom by replacing every circuit wire and gate by its respective gadget, and adding the necessary gadgets to satisfy excess pieces. Every circuit solution will have at least one Hexiom puzzle solution and every Hexiom puzzle solution has exactly one circuit solution. The solution from the puzzle to the circuit is taken from the variable assignments which can be extracted in linear time, achieving a polynomial-time reduction. As a result, Hexiom is NP-Hard.
    
    We show that Hexiom $\in$ NP by the existence of a polynomial time verification algorithm. Given a board and every piece's position, we iterate through every numbered piece, checking the local constraint. Each piece only need to evaluate its six neighbours, for a total complexity of $O(|F \cup M| \times 6)$, which is polynomial.
    \qed \end{proof}

\section{Cut the Rope}\label{proof:cuttherope}
    Cut The Rope is a very popular physics based mobile game developed by Zepto Lab and first released in 2010. It has since spawned many sequels and variations. Here, we prove that the original game is \textbf{NP}-Hard, using the 3-CNFSAT framework presented in Section \ref{fw:cnf}.
    
    \subsection{Rules}
    The objective of the game is to carry a candy (avatar) to Om Nom's mouth (finish line). The player does not control the candy directly, however. The levels have a variety of physics-based objects with which the player interacts to change the candy's trajectory. The relevant game rules for this proof are the following:
    \begin{itemize}
        \item Ropes (Figure \ref{fig:cuttherope_rope}) are initially deactivated and trigger (become activated, attached to the candy) when the candy enters their radius, then acting as a pendulum. The player can swipe his finger on the rope to cut it. Ropes only trigger once, and are no longer usable once cut. We will refer to this as \textit{consuming} a rope. A rope's length is the size of its trigger radius.
        \item Movable ropes (Figure \ref{fig:cuttherope_moverope}) are ropes that the player can move around within the track. 
        \item Teleporting hats (Figure \ref{fig:cuttherope_hat}) come in pairs. If the candy enters one, it leaves through the other. The candy's velocity is conserved after teleporting, but the movement's direction changes with the hat's alignment.
        \item Spikes (Figure \ref{fig:cuttherope_spike}) break the candy on contact, making the player lose the level.
        \item Balloons (Figure \ref{fig:cuttherope_balloon}) can be pressed by the player to blow air. This applies force to the candy.
        \item Bubbles (Figure \ref{fig:cuttherope_bubble}) activate on contact and make the candy float. Once pressed, the bubbles pop and lose their effect for the remainder of the level.
    \end{itemize}
    
    It may be difficult to understand how the game works from these descriptions alone, so we encourage the reader to watch the game being played to get a more intuitive understanding of it.\footnote{https://www.youtube.com/watch?v=7HgssdcI-EM}
    
    \begin{figure}[!ht]
    \center
     \subfloat[Rope pivot.\label{fig:cuttherope_rope}]{%
       \includegraphics[scale=0.3]{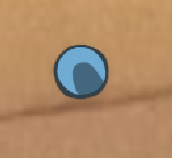}
     }
     \subfloat[movable rope.\label{fig:cuttherope_moverope}]{%
       \includegraphics[scale=0.3]{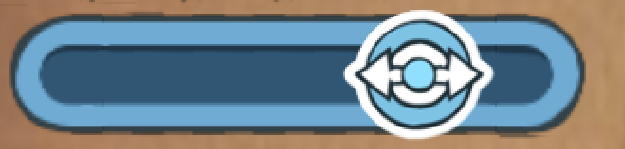}
     }
     \subfloat[Teleport hat.\label{fig:cuttherope_hat}]{%
       \includegraphics[scale=0.4]{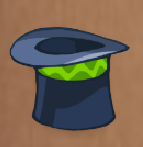}
     }
     \\
     \subfloat[Spikes.\label{fig:cuttherope_spike}]{%
       \includegraphics[scale=0.5]{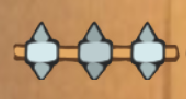}
     }
     \subfloat[Balloon.\label{fig:cuttherope_balloon}]{%
       \includegraphics[scale=0.3]{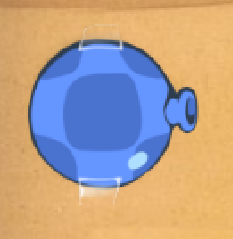}
     }
     \subfloat[Bubble.\label{fig:cuttherope_bubble}]{%
       \includegraphics[scale=0.4]{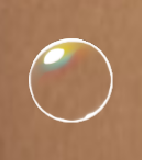}
     }
     \caption{Cut the Rope elements required for the proof}
     \label{fig:cuttherope_mechanics}
    \end{figure}
    
\subsection{Reduction from 3-CNFSAT}\label{fw:cnf}
Here we present the reduction from 3-CNFSAT given by Demaine et. al\cite{nintendopaper}, refining Vigilietta~\cite{forisek}. 
\begin{definition}
3-CNFSAT is the problem of determining whether a boolean formula, in 3 - Conjunctive Normal Form, has a variable assignment that makes its output $true$.
A 3-CNF formula is a conjunction ($\land$) of clauses. Each clause is a disjunction ($\lor$) of at most three literals. A literal is a variable ($x$) or its negation ($\lnot{x}$).
\end{definition}
Figure \ref{fig:cnf_fw} shows the template for the reductions. This applies generically to games where there is an avatar controlled by the player, and the objective is to reach the end of the level (Super Mario Bros. being one example, presented in Demaine et. al\cite{nintendopaper}). The player avatar that begins in Start, and must reach Finish to successfully complete the level. The avatar's movement is restricted by the lines connecting each gadget. Aside from basic traversal (a characteristic shared between many video games), the gadgets required are Crossovers, Locks and One-way paths (see~\cite{viglietta} for alternative approaches), which we briefly describe below.
\begin{figure}[h!]
    \centering
    \includegraphics[scale=0.95]{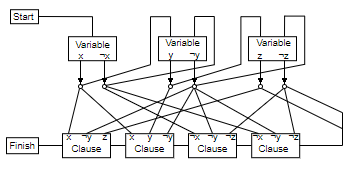}
    \caption{NP-Hardness framework from~\cite{nintendopaper}.}
    \label{fig:cnf_fw}
\end{figure}

\begin{definition}
A lock is a barrier that prevents the avatar from crossing it, unless its opening mechanism has been previously activated.
\end{definition}
The lock gadget is a mechanism that blocks a path, and can be unlocked, unblocking the path. It is used to implement clause satisfiability by associating one lock with each variable in said clause. The level must be setup so that lock $x$ is opened by choosing $true$ as $x$'s value. Lock $\lnot{x}$ remains locked. Each variable selection path (leaving each variable $x$, $y$, $z$ in Figure \ref{fig:cnf_fw}) is connected to the respective lock's opening mechanism. Every lock begins locked.

\begin{definition}
Given two lines that overlap in space, a crossover gadget enforces that the avatar must exit through the same path that he entered from.
\end{definition}The Crossover gadget is used whenever two of the lines cross, forcing the player to stay in the original wire. This effectively transforms the framework into a planar graph.

\begin{definition}
A one-way path is a path between two points, A and B, such that the avatar can go from A to B, but not from B to A.
\end{definition}
This is a simple gadget that prevents the player from traversing a line backwards. It is used immediately after the variable assignment, preventing the player from assigning it both values. A single-use paths can often be used instead of one-way paths.

\paragraph{Framework Overview}
The objective of the player is to go from the Start to the Finish position. The player is presented with a choice at each variable gadget, one leading to the line representing $true$ and the other to the line representing $false$.
Each clause is simply a room with three distinct locks, one for each variable. The clause gadget can be traversed if any of the three locks has been opened, just as a clause is satisfied by any of its three variables.
Each variable is connected to every clause in which it appears. The ``true line'' is connected to every clause where the variable appears non-negated, while the ``false line'' is connected to every clause where the variable appears negated. In the template of Figure \ref{fig:cnf_fw}, variable $x$ is connected to the opening mechanism of the first lock of two (counting from the left) clauses; variable $\lnot y$ is connected to the second lock of the first, third and fourth (counting from the left) clauses.

To successfully traverse the level, the player will have to find a satisfying assignment of the overall formula. The solution for the original 3-CNFSAT formula is directly retrieved from the sequence of lines (true or false) chosen at the variables. A more detailed explanation and proof of the framework's NP-Hardness can be found in Demaine et. al~\cite{nintendopaper} or in the thesis~\cite{mythesis}.
    \begin{figure}[!ht]
    \center
     \subfloat[Variable choice.\label{fig:cuttherope_variableselect}]{%
       \includegraphics[scale=0.37]{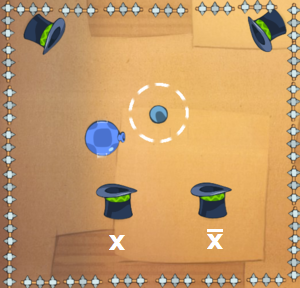}
     }
     \subfloat[Lock.\label{fig:cuttherope_lock}]{%
       \includegraphics[scale=0.28]{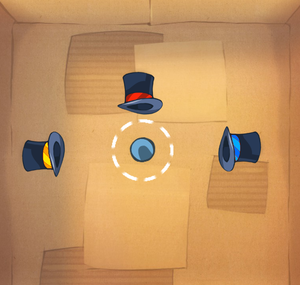}
     }
     \subfloat[Speed adjuster.\label{fig:cuttherope_speed}]{%
       \includegraphics[scale=0.235]{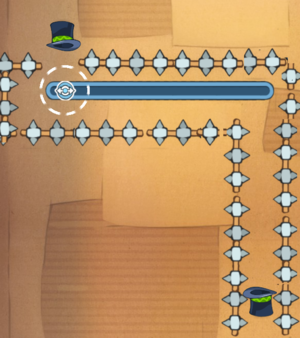}
     }
     \begin{center}
     \subfloat[Clause unlock.\label{fig:cuttherope_clausevisit}]{%
       \includegraphics[scale=0.5]{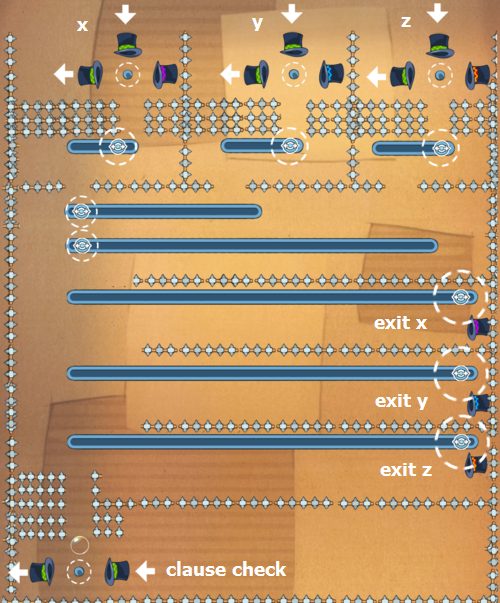}
     }
     \end{center}
     \hfill
     \caption{Cut the Rope gadgets}
     \label{fig:cuttherope_gadgets}
    \end{figure}

\subsection{Proof overview} For Cut the Rope, basic structure follows the framework closely, with three main differences: firstly, use of teleporting hats to simplify the lines and crossovers connecting each gadget; secondly, an additional gadget is also used to maintain an exact velocity magnitude when leaving the hats, controlling the candy's trajectory; thirdly, we used a single lock per clause, with a different opening mechanism for each variable.
Figure \ref{fig:cuttherope_gadgets} shows the gadgets required for the proof.

\subsection{Gadgets}
\paragraph{Lock} Figure \ref{fig:cuttherope_lock} implements the generic locking mechanism. We assume the rope is short enough so that it is impossible to reach any of the top hats by swinging on it. The gadget is traversed horizontally, but only after unlocking it from the top. If the player attempts to traverse the gadget without previously triggering the rope, the candy will be stuck in the rope, failing the level. Entering through the top, with the aid of a Bubble, the player will be able to trigger the rope, cut it, and then exit through the top, opening the lock. After this unlocking, the gadget can be traversed horizontally.

\paragraph{Velocity} Figure \ref{fig:cuttherope_speed} shows how to force an exact velocity when leaving teleport hats. The player enters through the top and leaves through the bottom. The shaft is made arbitrarily narrow to prevent the candy from having any horizontal speed when leaving the gadget. Under the constant force of gravity, and with any vertical speed lost when the rope is triggered, the candy's velocity is entirely determined by the length of this shaft.

\paragraph{Variable choice} Figure \ref{fig:cuttherope_variableselect} shows the variable selection gadget. The candy enters from one of the two hats at the top of the gadget, triggering the rope in the center. The player then uses the balloon to swing the candy, cutting the rope to direct it into one of the two hats (true or false).

\paragraph{Clause gadget} The clause gadget forms the bulk of the proof, shown in Figure \ref{fig:cuttherope_clausevisit}. The two hats at the bottom of the gadget represent the actual clause (using a lock) and its traversal path. The remainder of the gadget is there to guarantee the consistency of the gadget when unlocking the clause from different variables.

The general idea is that each variable has a different entry and exit path, illustrated by the hats' colors. Each exit path (at the top) starts locked, and can only be unlocked when the player enters the gadget through the top. This forces the player to leave only through variables that he entered from (having entered the gadget from $x$, the player is forced to leave through $x$ as well). The fact that every rope can be used only once ensures that each entry or exit path can also be used only once. The clause is unlocked by having the candy drop onto the bubble, triggering and cutting the rope (thus unlocking the clause), then letting it float upward onto its exit rope. For this to be correct, three properties must hold, described below. All three properties apply to the three variable locks at the top of the gadget, which are distinct from the clause's lock at the bottom of the gadget.

\begin{lemma}Property 1: The player must unlock the respective variable lock (consume the rope between the 3 hats on the top) to use the exit. \end{lemma}
\begin{proof} This follows from how the lock mechanism works. To traverse the path horizontally, the rope must have been previously consumed, or the candy will get stuck and be forced to go down the gadget.
\qed \end{proof} 

\begin{lemma}Property 2: The player can't reach a variable lock from the path below and go back down through the same path.
\end{lemma}
\begin{proof} This could only be violated with the use of the bubble, floating all the way up to the top. However, to avoid the spikes and align the candy with the lock (from below), the topmost movable rope must be used. Note that the rope's range is short enough so that, to float up and trigger the lock, the player must cut the movable rope. After triggering the lock the player must burst the bubble and cut the rope to move back down. 
Doing so would let the candy fall onto the spikes, since the movable rope has already been consumed.
\qed \end{proof}
   
\begin{lemma}Property 3: The player can't reach the variable lock from the same colored hat twice.
\end{lemma}
\begin{proof} this follows from a similar principle to property 2's. To reach the colored hat, the candy must consume the colored hat's movable rope (to reach $x$'s purple hat at the top, the candy must consume the ``exit $x$'' rope). Doing this, the player could open an incorrect lock and move back down. However, because of the consumed rope, there is no way to reach the same colored hat again, and so the player will not be able to leave through the incorrect entry point.

From the three properties, we can conclude that the only way to traverse the clause safely is to unlock the variable lock by entering through the top, and then use the respective movable rope to exit the gadget. Note that the exit movable ropes can intercept the candy in its fall, so the clause does not \textit{need} to be unlock every time (or even any time). The player should then open it on the first clause visit, and simply traverse the gadget for the subsequent variables.
\qed \end{proof}

\begin{theorem}Cut the Rope is NP-Hard.\end{theorem}
\begin{proof} We show that 3-CNFSAT $\le _p$ CutTheRope by implementing the variable assignment and clause gadgets from Framework \ref{fw:cnf} with the respective CutTheRope gadgets. Wire endpoints are replaced with teleport hats, removing the need for crossover gadgets. Between any two teleport hats, we add the additional gadget in Figure \ref{fig:cuttherope_speed} to account for speed gain/loss during the gadgets' traversal. 
For every 3-CNFSAT formula, we create a CutTheRope level such that the level can be successfully finished if and only if the original formula is satisfiable. The variables' values between the level and the formula have a one-to-one correspondence, given by the path traversed in Figure~\ref{fig:cuttherope_variableselect}, where the left hat represents true and the right hat represents false. \qed \end{proof}

\section{Back to Bed}\label{proof:backtobed}
Back to Bed is a puzzle game where the player is tasked with escorting a sleep-walking character to the bed, preventing his falling into the abyss or being caught by patrolling dogs.

We prove that this game is PSPACE-Hard with a reduction from TQBF, using the framework presented in Section \ref{fw:tqbf}. The game features teleporting mirrors, making the wiring easier and removing the need for a crossover gadget. 

\subsection{Rules} There are two characters: the player avatar, and the sleep-walker; there are also two objects the player can pick up and place down: apples and fish bridges; finally, there are two kinds of obstacles: patrolling dogs, and holes. Game levels are built out of solid tiles, walls, and teleporting mirrors. The sleep-walker always walks forward, turning ninety degrees to the right when it collides with an apple or a wall. The avatar can be moved at will (over solid tiles) and is able to carry one of two objects at a time: an apple, or a bridge. Apples occupy a single tile, and require a solid tile below them. Bridges are four tiles long and can cover any gap of length $2$, but need a solid tile below its two endpoints. If the sleep-walker walks onto a hole (without a bridge), he will fall and the level will restart. The game also features dogs which walk forward and turn ninety degrees to the right when they collide with an obstacle, or a hole (thus, dogs never fall into the abyss). Mirrors always appear in pairs teleport, and when the avatar or sleep-walker walk into one of the mirrors, they are teleported to the other.

\subsection{Reduction from TQBF}\label{fw:tqbf}
Here we present the reduction from PSPACE presented in~\cite{viglietta}, refining~\cite{forisek}, reducing from the True Quantified Boolean Formula (TQBF) problem.
\begin{definition}
TQBF is the problem of determining whether a boolean formula, with quantified variables (existentially and universally) is true.
A universally quantified variable must satisfy the formula when it is true and when it is false; an existentially quantified formula must satisfy the formula with either of its values.
Quantifiers are applied at the beginning of the formula. We are assuming the boolean formulas are also in 3-CNF.
\end{definition}

Figure \ref{fig:qbf_fw} shows the template for the reductions, with Figure \ref{fig:qbf_gadgets} showing the specific gadgets, implemented with Doors and their respective opening and closing mechanisms. Similarly to the previous framework, a Crossover gadget is typically required, but trivially solved with the use of teleporting mirrors; instead of One-way paths and Locks, only a Door gadget is required to reduce from TQBF.

\begin{figure}[h!]
    \centering
    \includegraphics[scale=0.85]{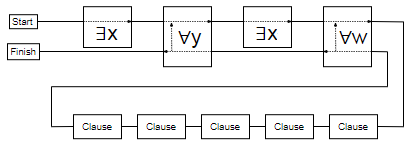}
    \caption{PSPACE-Hard framework from~\cite{nintendopaper}}
    \label{fig:qbf_fw}
\end{figure}

\begin{definition}\label{def:door}
A door is a barrier that prevents a character from crossing it when it is closed. Doors have two mechanisms: one to open it, and one to close it.
\end{definition}A door is similar to a lock, but it can be opened as well as closed. This property will be used to force the player to traverse the level with both values for universally quantified variables (the character must traverse the clauses with $x$ doors opened and $\lnot x$ door closed, and afterward with $x$ doors closed and $\lnot x$ doors opened). The player is always forced to close doors, but not to open them.

\subsection{Overview} This framework is similar to the previous one. In fact, 3-CNFSAT can be seen as a quantified boolean formula where every variable is existentially quantified. The introduction of universally quantified variables forces the player to visit variables several times, increasing the complexity to PSPACE instead of NP.

As in the previous framework, the objective is for the player to go from Start to Finish.

To do this, the player must go through the sequence of clauses (as well as the variables) multiple times, because of the way that universal variables can be traversed. Looking at Figure \ref{fig:qbf_fw}, we can see that universal variables have two entry points: one from the top, one from the bottom. They also have a path from bottom to top. When entering from the top, the bottom \textit{exit} is closed, forcing the player to traverse the gadget from the bottom to top, which then opens the bottom exit. How this is enforced with the use of Doors is shown in Figure \ref{fig:qbf_uni} and explained later.

Existential variables are similar to the variables in the previous framework. One of its two values is chosen in each traversal. A schematic of the gadget, as a function of doors, is shown in Figure \ref{fig:qbf_exi}. 

Figure \ref{fig:qbf_gadgets} shows the clause and variable gadgets implemented through the use of Door gadgets. The light gray zones represent the level's solid tiles. Dark gray tiles represent Doors. Tiles labelled with $+$ represent a Door's opening mechanism, while tiles labelled with $-$ represent its closing mechanism. Doors $a, b, c, d$ are local to each gadget, while clause doors $D_0, D_1, D_2$ are more global, being opened or closed by the variable gadgets' traversal.

\paragraph{Clauses}
A clause, shown in Figure \ref{fig:qbf_clause}, is simply a fragment of a level that can be traversed through any of three paths, with a Door in each path. The player should choose to traverse through whichever Door is open. If no door is open, either the player must have chosen an incorrect assignment of the variables, or the formula is false.

\paragraph{Existentially quantified variables} 
The existentially quantified variable, shown in Figure \ref{fig:qbf_exi}, can be traversed in either of two paths. The top path represents an assignment of $true$, while the bottom path an assignment of $false$. When traversing the top path, every Door of the clauses in which the variable appears with a positive value ($x_0$, $x_1$, etc.) are opened, while those in which it appears negatively valued ($\lnot x_0$, $ \lnot x_1$, etc.) are closed. In the bottom path it is the other way around. 

\paragraph{Universally quantified variables} The universally quantified variable, shown in Figure \ref{fig:qbf_exi}, contrary to the previous one, gives no choice to the player. Upon first entering the gadget, from the top left, the player is forced to assign the value true (by analogy with the previous gadget's sequence of opening and closing mechanisms), as door $c$ is closed, leave through the top right. Upon returning, through the bottom right, the player is forced to assign the value false, because door $d$ is closed. While choosing the value $false$, door $b$ is closed (preventing backtracking) and door $d$ is opened, allowing the player to leave the variable after a second return. We can see that, to successfully traverse the level resulting from this construction, the clause must always be satisfied by both values of a variable.

\begin{figure}[!ht]
 \subfloat[TQBF clause gadget.\label{fig:qbf_clause}]{%
   \includegraphics[scale=0.4]{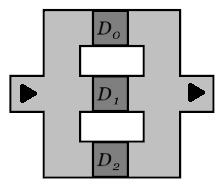}
 }
 \hfill
 \subfloat[TQBF existential quantifier gadget for x.\label{fig:qbf_exi}]{%
   \includegraphics[scale=0.5]{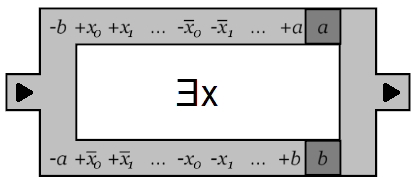}
 }
 \hfill
 \center
 \subfloat[TQBF universal quantifier gadget for x.\label{fig:qbf_uni}]{%
   \includegraphics[scale=0.5]{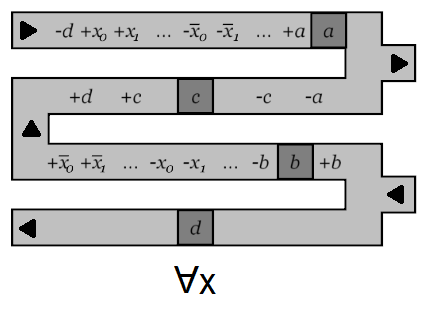}
 }
 \hfill
 \caption{TQBF gadgets from~\cite{nintendopaper}.}
 \label{fig:qbf_gadgets}
\end{figure}

\subsection{Proof overview} We now describe how we implement each gadget using Back to Bed's rules. The most important gadgets are shown in Figure \ref{fig:backtobed_gadgets}, with additional connecting gadgets in Figure \ref{fig:backtobed_forkjoin}. White tiles represent holes. The light blue rectangles show where the bridge can be placed. Red and green rectangles show where dogs can be patrolling. We implement doors using the patrolling dogs, as shown in Figure \ref{fig:backtobed_door}. The dog patrols either the red or the green path, one implementing an open door, the other a closed one. Placing the bridge between the two paths allows the dog to move between them.

Apples are used to implement any choice in the path, letting the player control redirect sleep-walker. We add an additional gadget (Figure \ref{fig:backtobed_pathseg}) to prevent apples from being moved to any gadget other than their starting one. 
The player is assumed to always be carrying a bridge (temporarily placing down it to carry apples or to help the sleep-walker cross gaps). The level is punctuated with holes, forcing the player to escort the sleep-walker closely, placing the bridge to prevent the sleep-walker from falling.

\begin{figure}[!ht]
\center
 \subfloat[Existential gadget.\label{fig:backtobed_existential}]{%
   \includegraphics[scale=0.4]{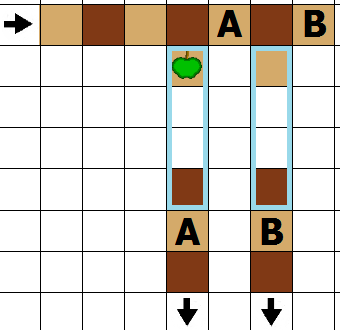}
 }\hfill
 \subfloat[Clause gadget.\label{fig:backtobed_clause}]{%
   \includegraphics[scale=0.4]{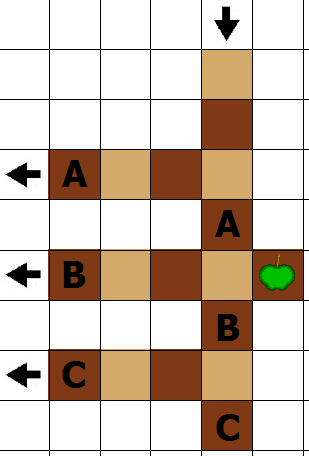}
 }\\
 \subfloat[Path-Segment gadget.\label{fig:backtobed_pathseg}]{%
   \includegraphics[scale=0.4]{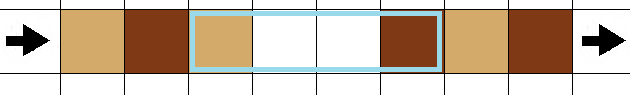}
 }\\
 \subfloat[Extending Path-Segment gadgets.\label{fig:backtobed_pathseg2}]{%
   \includegraphics[scale=0.35]{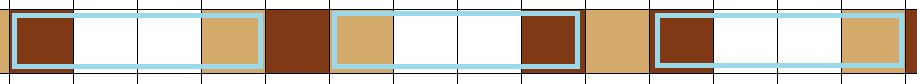}
 }\\
 \subfloat[Door gadget.\label{fig:backtobed_door}]{%
   \includegraphics[scale=0.5]{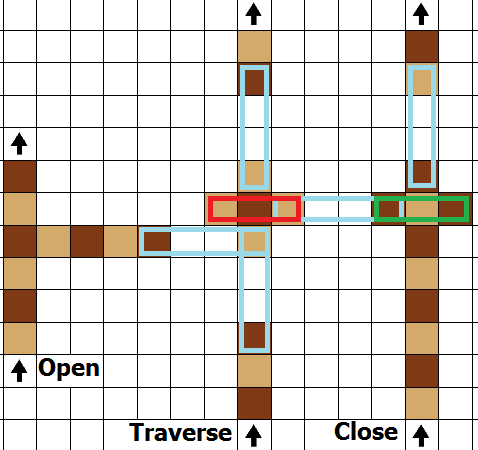}
 }
 \caption{Back to Bed gadgets}
 \label{fig:backtobed_gadgets}
\end{figure}

\subsection{Gadgets}
\paragraph{Door gadget (Figure~\ref{fig:backtobed_door})} This is the most important gadget in this reduction, as it is what makes the game PSPACE-Hard, instead of merely NP-Hard.

The door's state is given by the dog's patrolling area. As stated before, dogs always walk forward, until they encounter an obstacle or gap, at which point they turn 90 degrees clock-wise. In our reduction, each door has one dog, which may be on the green or red path. When the dog is patrolling on the green path, the door is open (as it leaves the Traverse path unobstructed); conversely, when the dog is patrolling the red path, the door is closed (as its Traverse path is obstructed).

Because no apples may be carried into this gadget, the only way to interact with the dog is through bridges. A bridge between the red and green paths will allow the player to change the door's state by waiting for the dog cross the bridge, thus switching from a red to a green patrolling path, or vice-versa.

\begin{lemma}The gadget in Figure \ref{fig:backtobed_door} implements a Door gadget (Definition \ref{def:door}), which can be opened and closed; the traverse path can only be traversed when the door is opened.\end{lemma}

\begin{proof} There are three paths in the gadget. Open, Close and Traverse. The dog is walking back and forth in either the red or green paths. If it is on the green path, the door is open; if it is on the red path, the door is closed.

In the \textit{Open} path, the player may \textit{choose} to change the door's state, or not. The sleep-walker will simply traverse in a straight line, bottom to top. This gives the player time to take the horizontal path and use the bridge to change the dog's state, if they so choose.

Traversing the \textit{Close} path \textit{always leaves the door closed}. For the sleep-walker to walk through the middle section, the dog cannot be patrolling the green area, or it would collide with the dog, resulting in a level restart. This forces the player to close the door (if it is open, with the dog patrolling the green path) by bridging the gap until the dog crosses it and patrols the red path, and then removing the bridge and using it for the sleep-walker to walk across the gap. Note that the beginning of the path gives the player enough time to change the dog's path from green to red (using the bridge), if needed.

The \textit{Traverse path} does not change the door's state, but it requires it to be open beforehand (dog patrolling the green zone) to be traversed. The gap in the path prevents the player from using the bridge to change the dog's state.
\qed \end{proof}

\paragraph{Path Segment gadget (Figure ~\ref{fig:backtobed_pathseg})} This gadget has two purposes: to keep the player close to the sleep-walker (thus preventing him from changing the game's consistency), and to prevent apples from being moved from one gadget to another. A constant number of these gadgets is added between any two mirrors.

\begin{lemma} The gadget shown in Figure \ref{fig:backtobed_pathseg} forces the player to stay close to the sleep-walker.\end{lemma}
\begin{proof} The player must move the bridge while the sleep-walker is traversing the center square, or it will inevitably fall down. The traversal of a single square does allow enough time for the player to move to a different gadget or change a door's state, and so must always follow the sleep-walker closely to place bridges before the sleep-walker falls.\qed\end{proof}

\begin{lemma} The gadget shown in Figure \ref{fig:backtobed_pathseg} prevents apples from being carried between gadgets.\end{lemma}
\begin{proof} As stated in the rules, the player can only carry one item at a time. Apples can't be placed on squares occupied by the sleep-walker or a bridge. While the player is crossing the path segment, there is no empty room for an apple to be placed, since two of the squares (at both extremes) will be occupied by the bridge (if the bridge is on the leftmost square, the rightmost square must be clear and vice-versa), and the center one by the sleep-walker. The time taken for the sleep-walker to traverse a single square is also short enough that the player doesn't have enough time to keep the apple in the segment immediately behind the sleep-walker's. If the current segment is $i$, the apple would be in $i-2$. The player would then need to move the bridge from $\langle i-1, i \rangle$ to $\langle i-2, i-1\rangle$, then grab the apple, moving it from $i-2$ to $i-1$, then move the bridge back to $\langle i-1, i\rangle$, and finally to $\langle i, i+1\rangle$ for the sleepwalker to traverse. The distance between the center of $i$ and $i-2$ is of $10$ squares. Even assuming instantaneous bridge placement, the player would need to traverse at least $15$ squares while the sleep-walker traverses a single one. This makes it impossible for an apple to be carried between path segments.
\qed \end{proof}

\paragraph{Fork and Join} To finish the reduction, we show how to fork and join paths to implement choice. Figure \ref{fig:backtobed_forkjoin} illustrates the gadgets. With these two, it is easy to implement the choice present in existential and universal gadgets from Section \ref{fw:tqbf}. All the remaining elements in the framework are straightforward connections using the teleporting mirrors.

\begin{figure}[!ht]
\center
     \subfloat[Fork.\label{fig:backtobed_fork}]{%
       \includegraphics[scale=0.31]{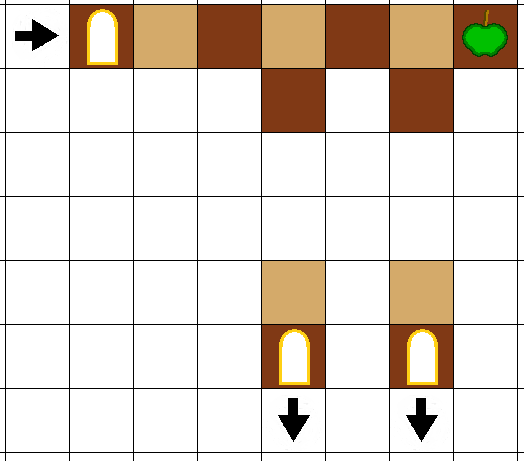}
     }
     \hfill
     \subfloat[Join.\label{fig:backtobed_join}]{%
       \includegraphics[scale=0.5]{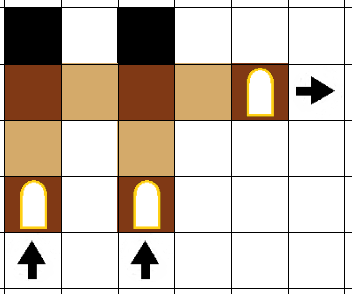}
     }
     \caption{Fork and Join auxiliary gadgets}
     \label{fig:backtobed_forkjoin}
\end{figure}

\begin{theorem}Back to Bed is PSPACE-Hard.\end{theorem}
\begin{proof} Using the framework described in \ref{fw:tqbf} and the gadgets in \ref{fig:backtobed_gadgets}, we show that TQBF $\le _P$ Back to Bed, thus proving the game to be PSPACE-Hard. In TQBF there is no verification necessary, reaching the end of the level is proof by itself that the quantified formula is true.
\qed \end{proof}

\subsection{Example instances} For brevity, we omit the examples of specific instances in this paper. In the main thesis~\cite{mythesis}, every proof has an example in the appendix, showing a formula (or circuit) and the level resulting from the reduction. It may help the reader to understand how gadgets fits together.

\section{Related Work}\label{section:rw} Part of this paper closely follows a series of papers on the complexity of video games initiated by Fori\v sek (2010)~\cite{forisek} and continued by Viglietta (2012)~\cite{viglietta} and Demaine (2012)~\cite{nintendopaper}. Specifically, the 3-CNFSAT and TQBF frameworks were taken from~\cite{nintendopaper}. In this paper, we reduce from 3-CNFSAT to Cut the Rope and from TQBF to Back to Bed.

The reduction from CircuitSAT to pen and pencil puzzle games was first approached by Friedman(2000), proving Spiral Galaxies~\cite{spiral} (published by Nikoli) NP-Complete. The same technique was later used by Brandon McPhail to prove Akari NP-Complete\cite{akari} when played on the standard square grid, another Nikoli pen-and-pencil puzzle game. In this paper, we apply these techniques to show that Hexiom is NP-Complete; in the thesis~\cite{mythesis}, we also extend McPhail's results to different grids (the two other euclidean regular tilings, hexagonal and triangular).

\section{Conclusion}
In this paper we showed that Hexiom and Cut the Rope are NP-Hard, in the process we show several insights about these games. Cut the Rope is a physics based game, hence the gadgets we use are unsual, in particular we need to handle physics problems, such as speed adjusting. Angry birds is another physics based games which was recently shown to be NP-Hard~\cite{angrybirds}. Hexiom is a puzzle game, however a straight forward construction of local constrains to simulate circuit components is tricky. We obtain such a reduction by uncovering the solution's global parity property~\ref{proof:hexiom}. Back to bed is a unique game in the sense that the player's work is to coordinate the avatar's movement. For this game we obtain the stronger result that it is PSPACE-Hard, by using the framework of~\cite{viglietta}. We build door gadgets, that need to be re-usable.

This paper proves the NP-Hardness of Hexiom and Cut the Rope as well as the PSPACE-Hardness of Back to Bed. Using well established frameworks, we show their applicability to modern games, further extending the previous results that NP and PSPACE-Hardness are widespread in games. Cut the Rope and Back to Bed's proofs apply the frameworks to games with novel movement systems. Hexiom's proof shows a mixture of local constraints to simulate circuit components and a global parity constraint to guarantee solutions to the puzzles.

\section*{Acknowledgements}
The work reported in this article was supported by national funds through
Funda\c{c}\~ao para a Ci\^encia e Tecnologia (FCT) through projects NGPHYLO
PTDC/CCI-BIO/29676/2017 and UID/CEC/50021/2019. Funded in part by European
Union’s Horizon 2020 research and innovation programme under the Marie
Sk{\l}odowska-Curie Actions grant agreement No 690941.

\bibliographystyle{plain}
\bibliography{references}

\end{document}